\documentclass{amsart}
\usepackage{amsmath, amssymb, amsthm}
\usepackage{graphicx, enumitem, calc, bbm}
\newtheorem{theorem}{Theorem}[section]
\newtheorem{lemma}[theorem]{Lemma}
\newtheorem{proposition}[theorem]{Proposition}

\theoremstyle{definition}
\newtheorem{definition}[theorem]{Definition}

\theoremstyle{remark}
\newtheorem{remark}[theorem]{Remark}
\numberwithin{equation}{section}
\newcommand{\mysep}{\hspace*{5pt}}
\newcommand{\mylist}[3]
{%
\begin{trivlist}%
\item
\begin{list}{}{\setlength{\leftmargin}{2em}\setlength{\labelwidth}{\leftmargin}}%
\item[(1)] #1
\item[(2)] #2
\item[(3)] #3
\end{list}%
\end{trivlist}%
}
\begin{document}
\title{Perfect Codes in the Discrete Simplex}
\author{Mladen Kova\v{c}evi\'{c}}
\address{Department of Electrical Engineering, University of Novi Sad,
         Trg Dositeja Obra\-dovi\'{c}a 6, 21000 Novi Sad, Serbia}
\email{kmladen@uns.ac.rs   ({\bf corresponding author}),  dejanv@uns.ac.rs}
\author{Dejan Vukobratovi\'{c}}
%
%
\date{October 10, 2013.}
\begin{abstract}
We study the problem of existence of (nontrivial) perfect codes in
the discrete $ n $-simplex
$ \Delta_{\ell}^n := \left\{ \begin{pmatrix} x_0, \ldots, x_n \end{pmatrix} :
x_i \in \mathbb{Z}_{+}, \sum_i x_i = \ell \right\} $
under $ \ell_1 $ metric. The problem is motivated by the so-called
multiset codes, which have recently been introduced by the authors
as appropriate constructs for error correction in the permutation
channels. It is shown that $ e $-perfect codes in the $ 1 $-simplex
$ \Delta_{\ell}^1 $ exist for any $ \ell \geq 2e + 1 $, the
$ 2 $-simplex $ \Delta_{\ell}^2 $ admits an $ e $-perfect code if
and only if $ \ell = 3e + 1 $, while there are no perfect codes in
higher-dimensional simplices. In other words, perfect multiset codes
exist only over binary and ternary alphabets.
\end{abstract}
\maketitle
\section{Introduction}
\par The study of perfect codes is a classical, and perhaps one of the
most attractive topics in coding theory. The best studied case are
certainly codes in the Hamming metric spaces
\cite{sloane, cohen, vanLint, tietavainen, zinoviev, best, etzion_vardy},
as they are historically the first codes that were introduced and are
most relevant in practice. There are various other interesting examples
in the literature, however, such as perfect codes under the Lee metric
\cite{astola, horak3, etzion_lee, golomb_welch, horak1, horak2, spacapan},
Levenshtein metric \cite{levenshtein_perf, bours}, codes in projective
spaces \cite{etzion_proj}, Grassmanians \cite{grass1, grass2}, etc.
Delsarte's conjecture \cite{delsarte} on the non-existence of perfect
constant-weight codes under the Johnson metric has also inspired a lot
of research, and still remains unsolved
\cite{roos, etzion_siam, shimabukuro, etzion_schwartz, gordon, etzion_des}.
Many of these problems can be regarded as particular instances of the general theory of
perfect codes in distance-transitive graphs \cite{biggs} (but not all
cases of interest fit into this framework).
In the present paper we investigate perfect codes in discrete simplices
of arbitrary dimension. As discussed in Section \ref{motivation}, codes
in such spaces arise naturally in the context of error correction in the
so-called permutation channels.

The paper is organized as follows. The basic concepts used in
the sequel are introduced in the following subsection. Subsection
\ref{mainresults} summarizes the main contributions of the paper.
Section \ref{motivation} explains the motivation for studying codes
in discrete simplices; the notion of multiset codes is introduced
here and some of their properties are established. Proofs of the
results are given in Section \ref{proofs}.
\subsection{Notation and terminology}
\label{notation}
Let $ \mathbb{Z}_{+} = \{ 0, 1, 2, \ldots \} $ denote the set of nonnegative
integers. Let $ ( S , d ) $ be a finite metric space with an integer-valued
metric $ d $, and $ { \mathcal C } \subseteq S $ an error-correcting code.
Elements of $ \mathcal C $ are called codewords in this context.
\begin{definition}
 $ { \mathcal C } $ is said to be $ e $-perfect, $ e \in \mathbb{Z}_{+} $,
if balls of radius $ e $ centered at codewords are disjoint and cover
the entire space:
\begin{equation}
  { \mathcal B }(x,e) \cap { \mathcal B }(y,e) = \emptyset
  \quad  \text{for every}  \quad  x, y \in { \mathcal C }, \,\, x \neq y ,
\end{equation}
and
\begin{equation}
  \bigcup_{ x \in { \mathcal C } } { \mathcal B }(x,e) = S ,
\end{equation}
where $ { \mathcal B }(x,e) = \big \{ w \in S : d(x,w) \leq e \big\} $
is the decoding region of the codeword $ x $. In other words, every element
of $ S $ is at distance $ \leq e $ from exactly one codeword.
\end{definition}

Clearly, every singleton $ { \mathcal C } = \{ x \} $ is $ D $-perfect, with
$ D $ the diameter of the space $ S $, and $ S $ itself is $ 0 $-perfect.
In the rest of the paper, we shall be interested only in \emph{nontrivial}
perfect codes -- those with $ | { \mathcal C } | \geq 2 $ and $ e \geq 1 $.

Let $ n, \ell \in \mathbb{Z}_{+} $. The space under consideration
in this paper is the discrete version of the standard $ n $-simplex:
\begin{equation}
\label{eq_simplex}
  \Delta_{\ell}^n := 
    \left\{ \big( x_0, \mysep \ldots,\mysep x_n \big) \mysep : \mysep
             x_i \in \mathbb{Z}_{+}, \mysep \sum_{i=0}^n x_i = \ell \right\} ,
\end{equation}
endowed with the following metric:
\begin{equation}
\label{eq_metric}
  d(x,y) = \frac{1}{2} \lVert x-y \rVert_{_1}
         = \frac{1}{2} \sum_{i=0}^{n} | x_i - y_i | ,
\end{equation}
where $ x = ( x_0, \ldots, x_n ), y = ( y_0, \ldots, y_n ) $. (The constant
$ 1/2 $ is taken for convenience because $ \lVert x-y \rVert_{_1} $ is always
even for $ x, y \in \Delta_{\ell}^n $.) The diameter of $ \Delta_\ell^n $ under
$ d $ is clearly $ \ell $. Note that for $ x, y \in \Delta_{\ell}^n $
we can also write:
\begin{equation}
\label{eq_metric2}
  d(x,y) = \sum_{ x_i > y_i } (x_i - y_i) = \sum_{ x_i < y_i } (y_i - x_i) .
\end{equation}
To our knowledge, codes in this space have not been analyzed before.
Perfect codes under $ \ell_1 $ distance seem to have been studied only
in the integer lattice $ \mathbb{Z}^n $ (as periodic extensions of the
codes under the Lee metric), see e.g. \cite{golomb_welch, horak1, etzion_lee}.

It is particularly useful to represent the metric space
$ \left( \Delta_{\ell}^n, d \right) $ as a graph%
\footnote{In the graph
theoretic literature, $ 1 $-perfect codes are also known as efficient
dominating sets (see, e.g., \cite{domination}).}
with $ \left| \Delta_\ell^n \right| = { n+\ell \choose \ell } $ vertices,
and with edges connecting vertices at distance one. This representation
allows one to visualize the space under study, as well as codes in this
space, at least for $ n = 1, 2 $. Unfortunately, the resulting graph is
not distance-transitive and the general methods developed for such graphs
\cite{biggs} cannot be applied.%
\subsection{Main results}
\label{mainresults}
The following theorem summarizes the main contributions of the paper.
Its proof is deferred to Section \ref{proofs}.
\begin{theorem}
\label{thm_main}
Let $ e \geq 1 $. 
\mylist
 { Nontrivial $ e $-perfect code in $ \left( \Delta_{\ell}^1, d \right) $
   exists for every $ \ell \geq 2e + 1 $. Such a code has
   $ \big\lceil \frac{ \ell + 1 }{ 2e + 1 } \big\rceil $ codewords. }
 { Nontrivial $ e $-perfect code in $ \left( \Delta_{\ell}^2, d \right) $
   exists if and only if $ \ell = 3e + 1 $. Furthermore, there are exactly
   two such codes in $ \Delta_{3e + 1}^2 $, each having three
   codewords. }
 { Nontrivial $ e $-perfect code in $ \left( \Delta_{\ell}^n, d \right) $,
   $n\geq3$, does not exist for any $ e $ and $ \ell $. }
\end{theorem}

In addition to the existence proofs, we shall also enumerate in Section
\ref{proofs} all perfect codes in one- and two-dimensional simplices.
\section{Motivation -- Multiset codes}
\label{motivation}
\subsection{Permutation channel}
Let $ {\mathcal A} = \{ 0, 1, \ldots , n \} $ be a finite alphabet with
$ n + 1 \geq 2 $ symbols. A permutation channel over $ {\mathcal A} $ is
a communication channel that takes sequences of symbols from $ {\mathcal A} $
as inputs, and for any input sequence outputs a random permutation of
this sequence. Such channels arise, for example, in some types of packet
networks \cite{bertsekas} in which the packets comprising a single message
are routed separately and are frequently sent over different routes in
the network. Consequently, the receiver cannot rely on them being delivered
in any particular order.

In addition to random permutations, the channel is assumed to impose various
types of ``noise" on the transmitted sequence, such as insertions, deletions, 
and substitutions of symbols. For example, in a networking scenario mentioned
above, packet deletions can be caused by network congestion and consequent
buffer overflows in the routers, while packet substitutions (i.e., errors)
are usually caused by noise or malfunctioning of network equipment. Therefore,
the permutation channel with these types of impairments is indeed a relevant
model.
\subsection{Coding for the permutation channel}
It is clear from the definition of the permutation channel that, when
transmitting sequences through it, no information should be encoded in
the order of symbols in the sequence because it is impossible to recover
this information. The only carrier of information should be the
\emph{multiset} of the symbols sent, i.e., the number of occurrences of
each symbol from $ {\mathcal A} $ in the sequence. The appropriate space
in which error-correcting codes for the permutation channel should be
defined is therefore the set of all multisets over the channel alphabet
\cite{mladen_multiset}.

Formally, a multiset $ X $ is an ordered pair
$ ( {\mathcal A}, \mathbbm{m}_X ) $ where $ {\mathcal A} $ is the ground
set (the channel alphabet in our case) and
$ \mathbbm{m}_X : {\mathcal A} \to \mathbb{Z}_{+} $ is a
multiplicity function which encodes the numbers of occurrences of the
elements of $ \mathcal A $ in $ X $. The cardinality of $ X $ is the
number of elements it contains, including repetitions, namely
$ |X| = \sum_{i=0}^n \mathbbm{m}_X(i) $. Since the alphabet
$ {\mathcal A} = \{ 0, 1, \ldots , n \} $ is fixed, multisets can be
identified with their multiplicity functions, and these can be identified
with $ (n + 1) $-tuples
$ \big( \mathbbm{m}_X(0), \mysep \mathbbm{m}_X(1), \mysep \ldots , 
   \mysep \mathbbm{m}_X(n) \big) \in \mathbb{Z}_{+}^{n+1} $.
The set of all multisets over $ {\mathcal A} $ can therefore be identified
with the space $ \mathbb{Z}_{+}^{n+1} $, and the set of all
multisets of given cardinality $ \ell $ with
\begin{equation}
\left\{ \big( x_0, \mysep x_1, \mysep \ldots , \mysep x_n \big) \in \mathbb{Z}_{+}^{n+1}
        \mysep : \mysep \sum_{i=0}^n x_i = \ell \right\}
\end{equation}
which is precisely the discrete simplex $ \Delta_\ell^n $.
We shall focus here on the latter case only and study codes in $ \Delta_\ell^n $,
i.e., codes whose all codewords have the same cardinality. This convention
is also practically motivated because it somewhat simplifies the
communication protocol \cite{mladen_set, mladen_multiset}.
In order to study codes in $ \Delta_\ell^n $, it is convenient to introduce
a metric on this space so that the minimum distance of the code can be defined,
and guarantees on the number of correctable errors provided. A natural metric
on the space of multisets is the so-called symmetric difference metric
defined by:
\begin{equation}
  |X \bigtriangleup Y| = \sum_{i=0}^n |\mathbbm{m}_X(i) - \mathbbm{m}_Y(i)|
\end{equation}
where $ \bigtriangleup $ denotes the symmetric difference of (multi)sets.
This is obviously the $ \ell_1 $ distance between the sequences
$ \big( \mathbbm{m}_X(0), \mysep \ldots , \mysep \mathbbm{m}_X(n) \big) $
and $ \big( \mathbbm{m}_Y(0), \mysep \ldots , \mysep \mathbbm{m}_Y(n) \big) $.
We have thus shown that the metric space $ \left( \Delta_{\ell}^n, d \right) $
is an appropriate space for defining error-correcting codes in permutation
channels.

It should be pointed out that, even though we have defined codes in
$ \Delta_{\ell}^n $, the codewords from $ \Delta_{\ell}^n $ are not
what is actually sent through the permutation channel, they only describe
the multisets that are transmitted. Namely, if
$ \big(c_0, \mysep \ldots, \mysep c_n \big) $ is a codeword to be sent,
then what is actually transmitted are $ c_0 $ copies of the symbol $ 0 $,
$ c_1 $ copies of the symbol $ 1 $, etc. Therefore, for a multiset code
defined in $ \Delta_\ell^n $, $ \ell $ represents the length of the code
(the number of symbols in codewords = the cardinality of the codewords)
and $ n + 1 $ is the size of the alphabet.

\begin{remark}
 Note that in our setting the dimension of the code space depends on the
size of the alphabet ($ n + 1 $), not on the length of the code ($ \ell $).
This stands in sharp contrast with most other coding scenarios.
\end{remark}

\begin{remark}
 It is also interesting to observe that the Johnson space $ J(n+1, \ell) $
(the set of all binary sequences of length $ n+1 $ and Hamming weight $ \ell $)
is a subset of $ \Delta_\ell^n $. Furthermore, in the case of binary
sequences, the metric $ d $ from \eqref{eq_metric} reduces to the Johnson
metric.
However, (non)existence of perfect codes in $ \Delta_\ell^n $ does not imply
(non)existence of perfect codes in $ J(n+1, \ell) $, and the methods used
in this paper do not seem to be sufficient to settle Delsarte's conjecture
(except in some special cases that are already known, e.g., \cite[Cor 2]{etzion_schwartz}).
\end{remark}

Using the above terminology, we can now restate Theorem \ref{thm_main}
as follows:
\emph{
\mylist
 { Nontrivial $ e $-perfect multiset code of length $ \ell $ over a binary
   alphabet exists for every $ \ell \geq 2e + 1 $. Such a code has
   $ \big\lceil \frac{ \ell + 1 }{ 2e + 1 } \big\rceil $ codewords. }
 { Nontrivial $ e $-perfect multiset code of length $ \ell $ over a ternary
   alphabet exists if and only if $ \ell = 3e + 1 $. Furthermore, there are
   exactly two $ e $-perfect multiset codes of length $ 3e + 1 $, each having
   three codewords. }
 { Nontrivial $ e $-perfect multiset code of length $ \ell $ over a $ q $-ary
   alphabet, $ q > 2 $, does not exist for any $ e $ and $ \ell $. }
 }

Finally, we note that the framework presented in this section is a
generalization of coding in power sets \cite{gadouleau1, gadouleau2, mladen_set},
where codewords are taken to be \emph{sets} rather than multisets. Such
approaches to coding for the permutation channel are somewhat analogous
to the approach of K\"otter and Kschischang \cite{kk} of using codes in
projective spaces and Grassmanians for error correction in networks
employing random linear network coding. In both cases, the guiding idea
is to define codes in the space of objects invariant under the channel
transformation -- (multi)sets are invariant under permutations, whereas
vector spaces are invariant (with high probability) under random linear
combinations.
\section{Proofs}
\label{proofs}
We now proceed with the proof of our claim. To that end, it will be useful
to represent the simplex $ \Delta_\ell^n $ as the corresponding graph with
$ { n+\ell \choose \ell } $ vertices, and with edges connecting vertices at
distance one. As noted in Section \ref{notation}, such a representation will
allow us to visualize the spaces under study, at least in the case of binary
and ternary alphabets.
\subsection{Binary alphabet}
\label{proof_binary}
One-dimensional case is simple to analyze. The space
\begin{equation}
 \Delta_\ell^1 
   = \left\{ \big(\ell-t , \mysep  t)  \mysep : \mysep  t = 0, \ldots, \ell \right\}
\end{equation}
can be represented as a \emph{path} with $ \left| \Delta_\ell^1 \right| = \ell + 1 $
vertices, the leftmost vertex being $ \big( \ell, \mysep  0 \big) $ and
the rightmost $ \big( 0, \mysep  \ell \big) $ for example (see Figure
\ref{fig_binary}).

Since the diameter of $ \left( \Delta_{\ell}^1, d \right) $ is $ \ell $
and any two codewords of an $ e $-perfect code must be at distance
$ \geq 2e + 1 $, nontrivial code can exist only if $ \ell \geq 2e + 1 $.
It is not hard to conclude that a perfect code exists for any such $ \ell $
(see also \cite{domination} for the case $ e = 1 $).
Figure \ref{fig_binary} provides an illustration of such a code, and
Proposition \ref{thm_allperfbinary} lists all perfect codes in $ \Delta_\ell^1 $.
\begin{figure}[h]
\centering
   \includegraphics[width=0.7\columnwidth]{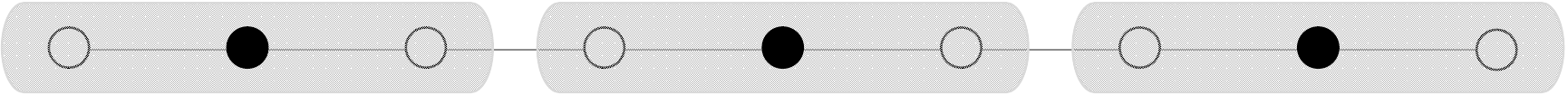}
 \caption{ $ 1 $-perfect code in $ \Delta_8^1 $ ($ n = 1 $, $ \ell = 8 $, $ e = 1 $).
           Black dots represent codewords; dots belonging to a gray region
           comprise the decoding region of the corresponding codeword. }
\label{fig_binary}
\end{figure}%
\begin{proposition}
\label{thm_allperfbinary}
 Let $ \ell = q(2e + 1) + r $ for some $ q \geq 1 $, $ 0 \leq r < 2e + 1 $.
Then there are exactly $ M = \min\{ r + 1, 2e + 1 - r \} > 0 $ perfect codes
in $ \Delta_\ell^1 $, each having
$ q + 1 = \big\lceil \frac{ \ell + 1 }{ 2e + 1 } \big\rceil $ codewords.
Let also $ s = \min\{ r, e \}$. Then all perfect codes in $ \Delta_\ell^1 $
can be enumerated as
\begin{equation}
\label{eq_allperfbinary}
 {\mathcal C}_1^{(m)} = 
   \Big\{ \big( \ell - s + m - 1 - i(2e + 1) , \mysep  
                s - m + 1 + i(2e + 1) \big)  \mysep :
          \mysep  i = 0, \ldots, q 
   \Big\} ,
\end{equation}
for $ m = 1, \ldots, M $.
\end{proposition}
\begin{proof}
 Considering the geometry of the space $ \Delta_\ell^1 $ and the corresponding
graph, it is clear that a perfect code has to be of the form
\begin{equation}
 \Big\{ \big( \ell - j - i(2e + 1) , \mysep  j + i(2e + 1) \big) \Big\} ,
\end{equation}
for some fixed $ j $, and for $ i $ ranging from $ 0 $ to some largest value.
Namely, once we have fixed the ``leftmost" codeword
$ \big( \ell - j, \mysep  j \big) $, all the other codewords are determined
by the fact that neighboring codewords have to be at distance $ 2e + 1 $
from each other. In that way we ensure that the decoding regions are disjoint
and that all intermediate points are covered. Therefore, to prove that
$ {\mathcal C}_1^{(m)} $ are perfect, i.e., that the entire $ \Delta_\ell^1 $
is covered, it is enough to show that the endpoints
$ \big( \ell , \mysep  0 \big) $ and $ \big( 0 , \mysep  \ell \big) $
are covered. Assume that $ r \leq e $, in which case $ M = r + 1 $ and $ s = r $.
Then $ 0 \leq s - m + 1 \leq r \leq e $, and hence the vertex
$ \big( \ell , \mysep  0 \big) $ is at distance $ \leq e $ from the codeword
$ \big( \ell - s + m - 1 , \mysep  s - m + 1 \big) $.
Similarly, $ 0 \leq r - s + m - 1 \leq r \leq e $ and therefore the vertex
$ \big( 0 , \mysep  \ell \big) $ is at distance $ \leq e $ from the codeword
$ \big( r - s + m - 1 , \mysep  \ell - r + s - m + 1 \big) $ (obtained for
$ i = q $ in \eqref{eq_allperfbinary}). Similar analysis applies when
$ r > e $. This proves that the codes $ {\mathcal C}_1^{(m)} $ are perfect.

It is left to prove that \eqref{eq_allperfbinary} lists all perfect codes
in $ \Delta_\ell^1 $. Assume that $ r \leq e $. In that case the ``leftmost"
codeword of $ {\mathcal C}_1^{(m)} $ is
$ \big( \ell - r + m - 1 , \mysep r - m + 1 \big) $, $ m = 1, \ldots, r + 1 $.
Therefore, we have found $ r + 1 $ codes with ``leftmost" codewords
$ \big( \ell , \mysep 0 \big) , \ldots, \big( \ell - r , \mysep r \big) $.
Suppose that we try to construct another perfect code by specifying
$ \big( \ell - r - k , \mysep r + k \big) $, $ k > 0 $, as its ``leftmost"
codeword. Since the end point $ \big( \ell , \mysep  0 \big) $ has to be
covered, we can assume that $ k \leq e - r $.
Then its ``rightmost" codeword is obtained by shifting for $ i(2e + 1) $
and is therefore either $ \big( 2e + 1 - k , \mysep \ell - 2e - 1 + k \big) $
(for $ i = q - 1 $) or $ \big( -k , \mysep \ell + k \big) $ (for $ i = q $).
The second case is clearly impossible, and the first fails to give a perfect
code because the point $ \big( 0, \mysep \ell \big) $ does not belong to a
decoding region of some codeword (its distance from the ``rightmost" codeword
is $ 2e + 1 - k > e $). Again, the proof is similar for $ r > e $.
\end{proof}
\subsection{Ternary alphabet}
\label{proof_ternary}
\par Consider now the two-dimensional simplex $ \Delta_{\ell}^2 $.
The graph representation of this space is a triangular grid graph, as
illustrated in Figure \ref{fig_2D_perfect} (we assume that the leftmost
vertex corresponds to $ \big( \ell, \mysep 0 , \mysep 0 \big) $, the
rightmost to $ \big( 0, \mysep \ell , \mysep 0 \big) $, and the top to
$ \big( 0, \mysep 0 , \mysep \ell \big) $). Balls under the metric $ d $ in
this graph are ``hexagons'', perhaps clipped if the center of the ball is
too close to the edge (in fact, this space is easily seen to be a ``triangle''
cut out from the hexagonal lattice, see Figure \ref{fig_2D_perfect}).
Hence, we need to examine whether a perfect packing of hexagons is
possible within this graph, i.e., whether there is a configuration of
hexagons covering the entire graph without overlapping. We first briefly
discuss some properties of $ \Delta_\ell^2 $ that will be useful.

Observe that, given some $ x \in \Delta_\ell^2 $, we can express
any point $ y \in \Delta_\ell^2 $ by specifying a path from $ x $ to
$ y $ in the corresponding graph. The first node on this path, call it
$ x' $, is a neighbor of $ x $, the second node is a neighbor of $ x' $,
etc. The neighbors of $ x = \big( x_0, \mysep x_1, \mysep x_2 \big) $,
i.e., points that are at distance $ 1 $ from it, are obtained by adding
$ 1 $ to some coordinate of $ x $, and $ -1 $ to some other coordinate.
A convenient way of describing neighbors and paths in $ \Delta_\ell^2 $
is as follows. Define the vector $ f_{i,j} $, $ i,j \in \{1,2,3\} $, to
have a $ 1 $ at the $ i $'th position, a $ -1 $ at the $ j $'th position,
and a $ 0 $ at the remaining position. For example,
$ f_{1,2} = \big( 1, \mysep -1 , \mysep 0 \big) $.
Clearly, $ f_{i,j} = - f_{j,i} $ and by convention we take
$ f_{i,i} = \big( 0, \mysep 0, \mysep 0 \big) $.
These vectors describe all possible directions of moving from some point,
and hence any neighbor $ x' $ of $ x $ can be described by specifying
the direction, namely $ x'= x + f_{i,j} $ (see Figure \ref{fig_directions}).
Therefore, any $ y \in \Delta_\ell^2 $ can be expressed as
\begin{equation}
 y = x + \sum_{i,j} \alpha_{i,j} f_{i,j}
\end{equation}
for some integers $ \alpha_{i,j} \geq 0 $.
\begin{figure}[h]
 \centering
    \includegraphics[width=0.5\columnwidth]{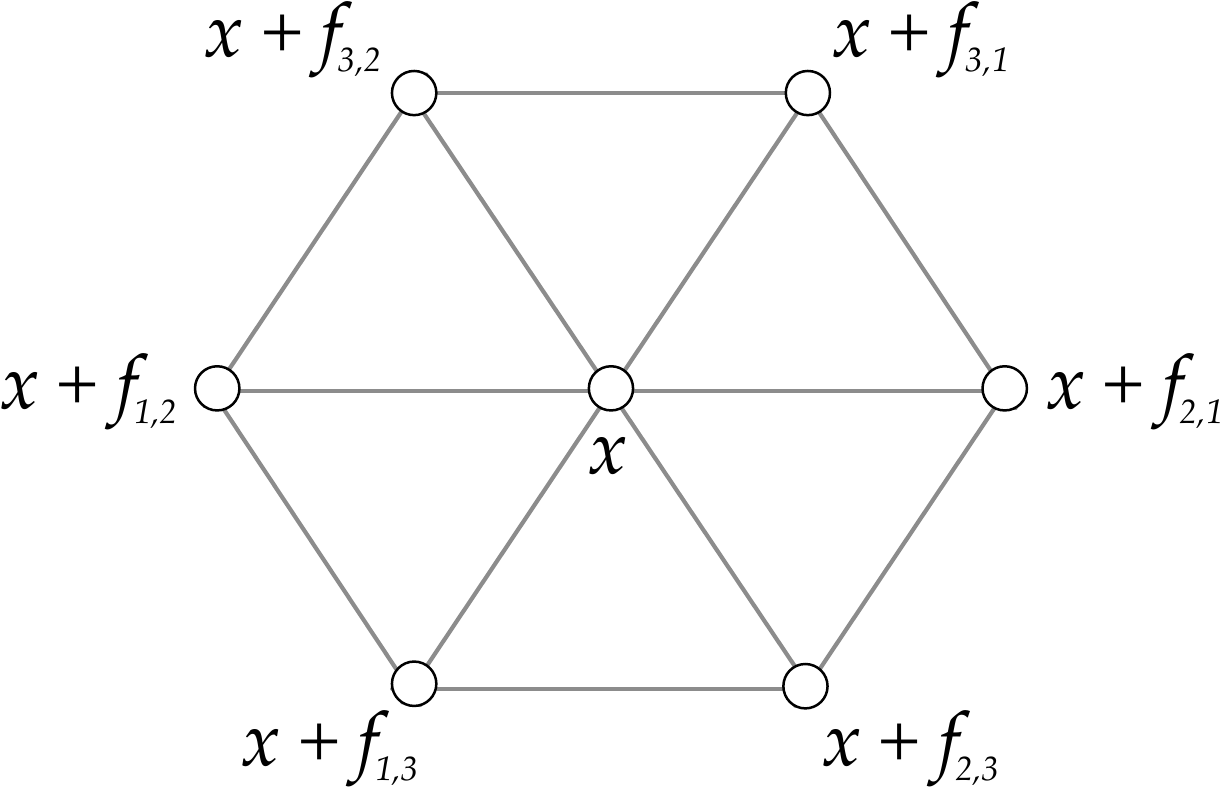}
  \caption{Neighbors of $ x $ in $ \Delta_\ell^2 $.}
\label{fig_directions}
\end{figure}%
If $ d(x,y) = \delta $, then clearly there exists a representation of this
form with $ \sum_{i,j} \alpha_{i,j} = \delta $. Another way to write this is
\begin{equation}
 y = x + \big( s_0, \mysep s_1, \mysep s_2 \big)
\end{equation}
where $ \sum_i s_i = 0 $ and $ \sum_i |s_i| = 2\delta $.

The following lemma will also be used in the sequel. The statement is
illustrated in Figure \ref{fig_2D_l00}, and the proof (of the
more general version) is given in the following subsection (see Lemma
\ref{prop_orthogonal} and Remark \ref{rem_orthogonal_2D}).
\begin{lemma}
\label{prop_orthogonal_2D}
 Let $ x, y, w \in \Delta_\ell^2 $ be such that $ d(x,w) = d(y,w) = e + 1 $,
$ d( x, w + f_{1,2} ) = e $, and $ d( y, w + f_{2,1} ) = e $. Then there can
be no $ z \in \Delta_\ell^2 $ such that $ w \in { \mathcal B }(z,e) $,
$ { \mathcal B }(x,e) \cap { \mathcal B }(z,e) = \emptyset $ and
$ { \mathcal B }(y,e) \cap { \mathcal B }(z,e) = \emptyset $.
\end{lemma}

Let us elaborate on the meaning of this lemma. Suppose we have two codewords
($ x, y $) and a point $ w $ lying outside their decoding regions. Since we are
trying to build a perfect code, the point $ w $ has to belong to a decoding
region of a third codeword $ z $. The lemma asserts that if $ w $ is bounded
by $ { \mathcal B }(x,e) $ and $ { \mathcal B }(y,e) $
in some direction, say $ f_{1,2} $ (recall that $ f_{2,1} = -f_{1,2} $),
then such a codeword cannot exist, and therefore $ x $ and $ y $ cannot be
codewords of a perfect code.

\begin{figure}[h]
 \centering
  \includegraphics[width=0.55\columnwidth]{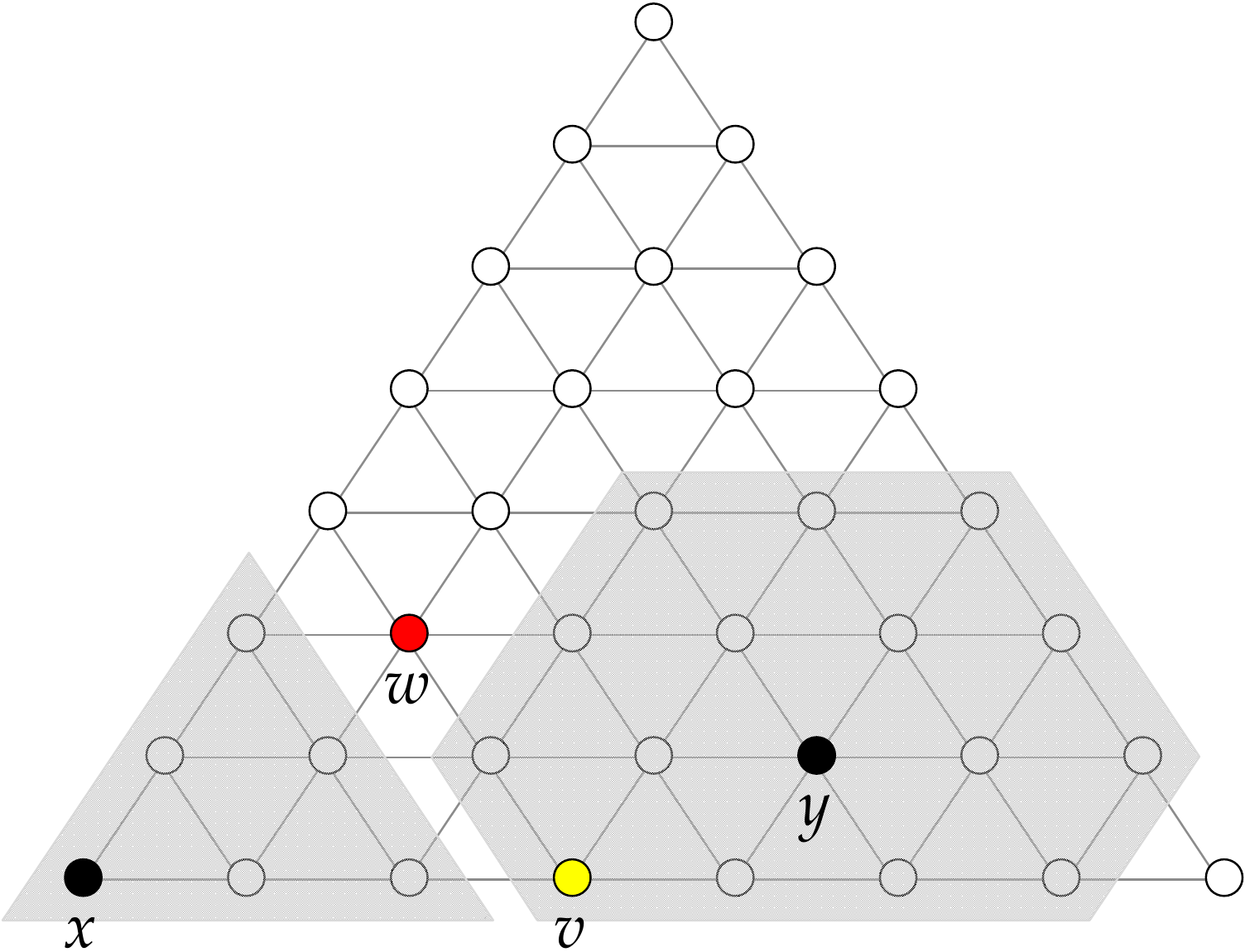}
  \caption{Illustration of Lemma \ref{prop_orthogonal_2D} and Lemma \ref{lemma_x_2D}.}
\label{fig_2D_l00}
\end{figure}%
%
%

We now proceed with proof of the main claim, namely the (non)existence
of perfect codes. If $ \ell = 3e + 1 $, then it is not hard to exhibit
a perfect code (see Figure \ref{fig_2D_perfect}).
In fact, there are exactly two such codes:
\begin{equation}
\label{eq_perfect_2D}
 \begin{aligned}
  {\mathcal C}_2^{(1)}  &= 
    \Big\{ \big( 2e + 1 , \mysep e , \mysep 0 \big) , \mysep
           \big( 0 , \mysep 2e + 1 , \mysep e \big) , \mysep
           \big( e , \mysep 0 , \mysep 2e + 1 \big)
    \Big\}  \\
  {\mathcal C}_2^{(2)}  &= 
    \Big\{ \big( 2e + 1 , \mysep 0 , \mysep e \big) , \mysep
           \big( e , \mysep 2e + 1 , \mysep 0 \big) , \mysep
           \big( 0 , \mysep e , \mysep 2e + 1 \big)
    \Big\} .
 \end{aligned}
\end{equation}
\begin{figure}[h]
\centering
 \includegraphics[width=0.55\columnwidth]{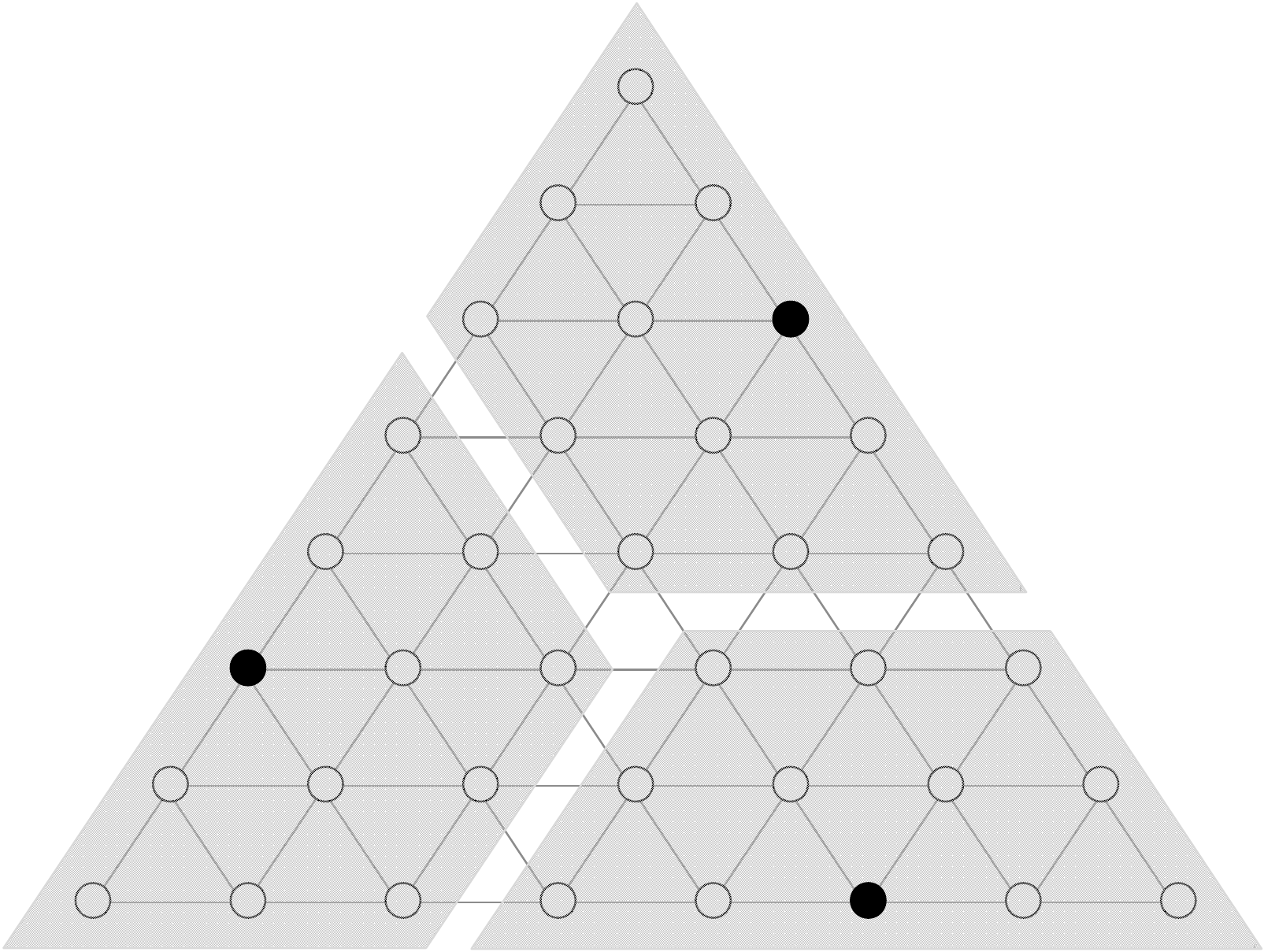}
 \caption{$ 2 $-perfect code ($ {\mathcal C}_2^{(2)} $) in
          $ \Delta_7^2 $ ($ n = 2 $, $ \ell = 7 $, $ e = 2 $).}
\label{fig_2D_perfect}
\end{figure}%
\begin{proposition}
 Codes $ {\mathcal C}_2^{(1)} $ and $ {\mathcal C}_2^{(2)} $ are $ e $-perfect
in $ \Delta_{3e+1}^2 $.
\end{proposition}

The proof of the proposition is straightforward and is omitted.
In the following we prove that these are the only two perfect
codes when $ \ell = 3e + 1 $, and that there are no perfect
codes for $ \ell \neq 3e + 1 $.

We start by observing the vertex $ \big( \ell , \mysep  0 , \mysep  0 \big) $.
For this vertex to be covered there must exist a codeword of the form
\begin{equation}
\label{eq_x_2D}
 x =  \big( \ell - t , \mysep  x_1 , \mysep  x_2 \big)
\end{equation}
with $ x_1 + x_2 = t \leq e $.
Observe now the point
\begin{equation}
\label{eq_v_2D}
 v =  \big( \ell - x_1 - e - 1 , \mysep  x_1 + e + 1 , \mysep  0 \big) .
\end{equation}
(Needless to say, we assume that $ v \in \Delta_\ell^2 $, i.e., that
$ v_0 = \ell - x_1 - e - 1 \geq 0 $; otherwise, the diameter of
$ \Delta_\ell^2 $ would be $ \ell \leq 2e $ and no nontrivial perfect
code could exist.)
We have $ d(x,v) = e + 1 $ and so the point $ v $ is not covered by
$ { \mathcal B }(x,e) $. To cover it we need another codeword $ y $
with $ d(v,y) = e $ and $ d(x,y) = 2e + 1 $.
\begin{lemma}
\label{prop_y_2D}
 Let $ x, v \in \Delta_\ell^2 $ be given by \eqref{eq_x_2D} and \eqref{eq_v_2D},
respectively. Then the point $ y \in \Delta_\ell^2 $ satisfying
$ d(v,y) = e $, $ d(x,y) = 2e + 1 $ is of the form
\begin{equation}
\label{eq_y_2D}
 y =  \big( \ell - x_1 - 2e - 1 , \mysep  x_1 + e + 1 + u , \mysep  e - u \big)
\end{equation}
with $ 0 \leq u \leq e $, and with the property that
\begin{equation}
\label{eq_yxzero_2D}
 x_2 > 0  \Rightarrow  u = e .
\end{equation}
\end{lemma}
\begin{proof}
 Let
$ y = \big( \ell - x_1 - 2e - 1 + s , \mysep  y_1 , \mysep  y_2 \big) $
for some $ s \in \mathbb{Z} $. If $ s < 0 $ we have
$ d(v,y) \geq v_0 - y_0 = e - s > e $ which contradicts one of the
assumptions of the lemma. We next show that the assumption $ s > 0 $ also
leads to a contradiction. We can assume that $ x_0 > y_0 $; otherwise, the
vertex $ \big( \ell , \mysep 0 , \mysep 0 \big) $ would be covered by both
$ x $ and $ y $. We can also assume that $ s \leq x_1 $, for otherwise
we would have $ x_0 - y_0 \leq 2e - t $, and since the sum of the
remaining $ x_i $'s is $ t $ it would follow that
\begin{equation}
\begin{aligned}
 d(x,y) &=     \sum_{ x_i > y_i }  ( x_i - y_i )
         =     x_0 - y_0 + \sum_{ i > 0,\, x_i > y_i }  ( x_i - y_i )  \\
        &\leq  x_0 - y_0 + \sum_{ i > 0 }  x_i
         \leq  2e .
\end{aligned}
\end{equation}
Now, since $ v_0 - y_0 = e - s < e $ and $ y_2 \geq v_2 = 0 $, we must have
$ v_1 - y_1 = x_1 + e + 1 - y_1 = s $ in order to achieve $ d(v,y) = e $
(see \eqref{eq_metric2}),
and hence
\begin{equation}
\label{eq_y1x1_2D}
 y_1 = x_1 - s + e + 1 \geq e + 1 > x_1 ,
\end{equation}
where the first inequality follows from the above assumption that
$ s \leq x_1 $. Since $ y_0 < x_0 $ and $ y_1 - x_1 = e + 1 - s $,
in order to have $ d(x,y) = 2e + 1 $ we must have $ y_2 - x_2 = e + s $.
But this is impossible because
\begin{equation}
 y_2 - x_2  \leq  y_2  =  \ell - y_0 - y_1  =  e  <  e + s ,
\end{equation}
where we have used \eqref{eq_y1x1_2D}. We thus conclude that $ s $
must be zero. In that case we have $ v_0 - y_0 = e $, and since
$ d(v,y) = e $, we must also have $ y_1 \geq v_1 = x_1 + e + 1 $.
This shows that $ y $ is necessarily of the form \eqref{eq_y_2D}.
To prove the last part of the claim observe that $ y_0 < x_0 $, 
$ y_1 - x_1 = e + 1 + u $, and $ d(x,y) = 2e + 1 $ imply that
$ y_2 - x_2 = e - u $ when $ u < e $. But since
$ y_2 = e - u $, this can only hold if $ x_2 = 0 $ whenever $ y_2 > 0 $.
\end{proof}

Assume therefore that we have two codewords of the form
\eqref{eq_x_2D} and \eqref{eq_y_2D}, and observe the point
\begin{equation}
\label{eq_w_2D}
  w = \big( \ell - t - e - 1 , \mysep  x_1 + u , \mysep  \max\{x_2, y_2\} + 1 \big) ,
\end{equation}
where $ y_2 = e - u $. (Here again we assume that $ w_0 \geq 0 $
because otherwise the diameter of $ \Delta_\ell^2 $ would be
$ \ell \leq 2e $.)
To show that $ w \in \Delta_\ell^2 $, consider two cases:
$ 1.) $  $ x_2 > 0 $; by \eqref{eq_yxzero_2D} this implies that
$ y_2 = e - u = 0 $ and $ \max\{x_2,y_2\} = x_2 $, wherefrom
$ \sum_i w_i = \ell $,
$ 2.) $  $ x_2 = 0 $; in this case $ t = x_1 $ and
$ \max\{x_2,y_2\} = y_2 = e - u $, so we again have $ \sum_i w_i = \ell $.
Furthermore, we have that $ d(x,w) = d(y,w) = e + 1 $. This is shown
easily by considering the above two cases. Namely, if $ x_2 > 0 $,
then $ y_2 = e - u = 0 $ and so
$ y =  \big( \ell - x_1 - 2e - 1 , \mysep  x_1 + 2e + 1 , \mysep  0 \big) $,
$ w = \big( \ell - t - e - 1 , \mysep  x_1 + e , \mysep  x_2 + 1 \big) $,
and by \eqref{eq_metric2} the statement follows.
The case $ x_2 = 0 $ is similar.

We shall need the following claim in the sequel (we omit the
proof because the statement is geometrically quite clear).
\begin{lemma}
\label{lemma_znogap}
 Let $ x , y , w \in \Delta_\ell^2 $ be such that $ d(x, w) = d(y, w) = e + 1 $,
$ d(x, w + f_{k,l}) = d(x, w + f_{m,l}) =  d(y, w + f_{k,m}) = e $. In words,
$ w $ is outside the decoding regions of $ x $ and $ y $, but its neighbors
along three consecutive directions (see Figure \ref{fig_directions}) are not.
Then the point $ z $ such that $ w \in {\mathcal B}(z, e) $,
$ {\mathcal B}(x, e) \cap {\mathcal B}(z, e) = {\mathcal B}(y, e) \cap {\mathcal B}(z, e) = \emptyset $
lies on the direction $ f_{l,k} = -f_{k,l} $, i.e., $ z = w + e f_{l,k} $.
\end{lemma}

\begin{lemma}
\label{lemma_x_2D}
 Let $ x, y \in \Delta_\ell^2 $ be given by \eqref{eq_x_2D} and \eqref{eq_y_2D}, 
respectively. Let also either $ a.) $ $ t < e $, or $ b.) $ $ t = e $ but
$ 0 < x_1 < e $. Then $ x $ and $ y $ cannot be codewords of an $ e $-perfect
code.
\end{lemma}
\begin{proof}
Assume first that $ x_2 > 0 $. Then, as noted above,
$ y = \big( \ell - x_1 - 2e - 1 , \mysep  x_1 + 2e + 1 , \mysep  0 \big) $,
$ w = \big( \ell - t - e - 1 , \mysep  x_1 + e , \mysep  x_2 + 1 \big) $.
Furthermore,
$ w + f_{1,2} = \big( \ell - t - e , \mysep  x_1 + e - 1 , \mysep  x_2 + 1 \big) $
and
$ w + f_{2,1} = \big( \ell - t - e - 2 , \mysep  x_1 + e + 1 , \mysep  x_2 + 1 \big) $.
By using \eqref{eq_metric2} we easily find that $ d(x,w) = d(y,w) = e + 1 $ and
$ d(x, w + f_{1,2}) = d(y, w + f_{2,1}) =  e $ (for the last equality we need
the fact that either $ t < e $, or $ t = e $ but $ x_1 > 0 $). Hence, by Lemma
\ref{prop_orthogonal_2D}, we conclude that there exists no codeword $ z $
whose decoding region contains $ w $ and is disjoint from the decoding regions
of $ x $ and $ y $.

Assume now that $ x_2 = 0 $. If $ u > 0 $, then
$ x = \big( \ell - x_1 , \mysep  x_1 , \mysep  0 \big) $
$ y = \big( \ell - x_1 - 2e - 1 , \mysep  x_1 + e + u + 1 , \mysep  e - u \big) $,
$ w = \big( \ell - x_1 - e - 1 , \mysep  x_1 + u , \mysep  e - u + 1 \big) $,
$ w + f_{1,2} = \big( \ell - x_1 - e , \mysep  x_1 + u - 1 , \mysep  e - u + 1 \big) $,
and
$ w + f_{2,1} = \big( \ell - x_1 - e - 2 , \mysep  x_1 + u + 1 , \mysep  e - u + 1 \big) $.
We therefore again have $ d(x,w) = d(y,w) = e + 1 $ and
$ d(x, w + f_{1,2}) = d(y, w + f_{2,1}) =  e $, and by Lemma \ref{prop_orthogonal_2D}
the conclusion follows.

Finally, if $ x_2 = 0 $ and $ u = 0 $, then
$ y = \big( \ell - x_1 - 2e - 1 , \mysep  x_1 + e + 1 , \mysep  e  \big) $,
$ w = \big( \ell - x_1 - e - 1 , \mysep  x_1 , \mysep  e + 1  \big) $,
$ w + f_{1,3} = \big( \ell - x_1 - e , \mysep  x_1 , \mysep  e  \big) $,
$ w + f_{2,3} = \big( \ell - x_1 - e - 1 , \mysep  x_1 + 1 , \mysep  e  \big) $,
and
$ w + f_{2,1} = \big( \ell - x_1 - e - 2 , \mysep  x_1 + 1 , \mysep  e + 1  \big) $.
Therefore, we have $ d(x,w) = d(y,w) = e + 1 $, $ d(x, w + f_{1,3}) = e $,
and $ d(y, w + f_{2,3}) = d(y, w + f_{2,1}) = e $. By Lemma \ref{lemma_znogap}
we conclude that the codeword $ z $ covering $ w $ has to be
$ z = w + ef_{3,2} = \big( \ell - x_1 - e - 1 , \mysep  x_1 - e , \mysep  2e + 1  \big) $,
but this is impossible because we have assumed that $ x_1 < e $ and therefore
the second coordinate of $ z $ is negative.
\end{proof}
The previous lemma shows that either $ \big( \ell - e , \mysep e , \mysep 0 \big) $
or $ \big( \ell - e , \mysep 0 , \mysep e \big) $ must be a codeword if the
vertex $ \big( \ell , \mysep 0 , \mysep 0 \big) $ is to be covered, and similarly
for the other two vertices $ \big( 0 , \mysep \ell , \mysep 0 \big) $ and
$ \big( 0 , \mysep 0 , \mysep \ell \big) $. This proves that the codes given by
\eqref{eq_perfect_2D} are the only perfect codes in $ \Delta_{3e + 1}^2 $.
It is left to prove that for $ \ell \neq 3e + 1 $ perfect codes do not exist.

\begin{proposition}
\label{prop_otherl}
 There are no $ e $-perfect codes in $ \Delta_\ell^2 $ for $ \ell \neq 3e + 1 $.
\end{proposition}
\begin{proof}
 The proof is illustrated in Figure \ref{fig_2D_general}, but we also give
here a more formal version. By the above arguments, we can assume that
$ x = \big( \ell - e , \mysep 0 , \mysep e \big) $ is a codeword. Observe
the point $ v = \big( \ell - e - 1 , \mysep e + 1 , \mysep 0 \big) $. By
Lemma \ref{prop_y_2D} we conclude that for $ v $ to be covered we must
take $ y = \big( \ell - 2e - 1 , \mysep 2e + 1 , \mysep 0 \big) $ to be
a codeword. Hence, we must have $ \ell \geq 2e + 1 $ for the perfect code
to exist. Now observe
$ w = \big( \ell - 2e - 1 , \mysep e , \mysep e + 1 \big) $.
We have $ d(x, w) = d(y, w) = e + 1 $ and so there must exist a third
codeword $ z $ covering $ w $. Note also that
$ d(x, w + f_{1,2}) = d(x, w + f_{1,3}) =  d(y, w + f_{2,3}) = e $ and so
by Lemma \ref{lemma_znogap} we conclude that $ z $ has to be of the form
$ w + e f_{3,1} $, i.e.,
$ z = \big( \ell - 3e - 1 , \mysep e , \mysep 2e + 1 \big) $.
Therefore, we must have $ \ell \geq 3e + 1 $ for the perfect code to exist.
The case $ \ell = 3e + 1 $ has been settled, so assume that $ \ell > 3e + 1 $.
Next, observe the point
$ u = \big( \ell - 3e - 2 , \mysep 2e + 1 , \mysep e + 1 \big) $.
We have $ d(z, u) = d(y, u) = e + 1 $ and $ d(x, u) = 2e + 2 $.
Therefore, to cover $ u $ we need a fourth codeword $ q $.
Since $ d(z, u + f_{1,2}) = d(z, u + f_{3,2}) = d(y, u + f_{1,3}) = e $,
by Lemma \ref{lemma_znogap} we conclude that
$ q = \big( \ell - 4e - 2 , \mysep 3e + 1 , \mysep e + 1 \big) $ (and so we
must have $ \ell > 4e + 1 $). Finally,
observe the point $ p = \big( \ell - 3e - 2 , \mysep 3e + 2 , \mysep 0 \big) $.
Its distance from the codewords $ x, y, z, q $ is easily seen to be $ > e $,
and therefore we need another codeword to cover it. However, since
$ d(q, p) = d(y, p) = e + 1 $ and
$ d(q, p + f_{3,1}) = d(q, p + f_{3,2}) =  d(y, p + f_{1,2}) = e $,
this codeword would (by Lemma \ref{lemma_znogap}) have to be
$ p + e f_{2,3} = \big( \ell - 3e - 2 , \mysep 4e + 2 , \mysep -e \big) $
which is impossible.
\end{proof}
\begin{figure}[h]
 \centering
  \includegraphics[width=0.65\columnwidth]{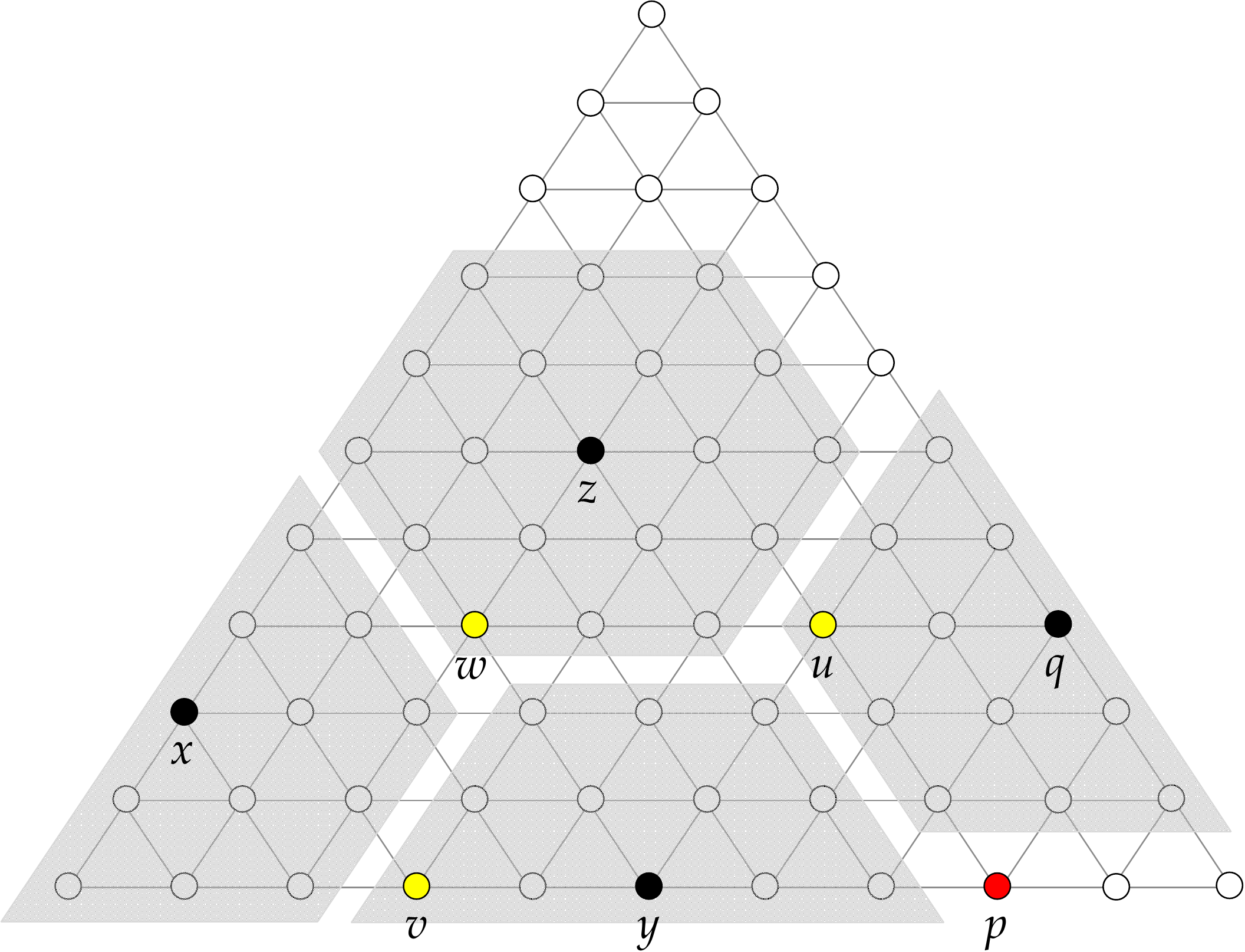}
  \caption{Proof of Proposition \ref{prop_otherl}.}
\label{fig_2D_general}
\end{figure}%
\subsection{Larger alphabets}
\label{proof_larger}

We now turn to the higher-dimensional case.

As in two dimensions, given some $ x \in \Delta_\ell^n $,
we can always express the point $ y \in \Delta_\ell^n $ by specifying a
path from $ x $ to $ y $. This is formalized by using vectors $ f_{i,j} $,
as before (the $ n $-dimensional vector $ f_{i,j} $ has a $ 1 $ at the
$ i $'th position, a $ -1 $ at the $ j $'th position, and zeros elsewhere, e.g.,
$ f_{1,2} = \big( 1, \mysep -1, \mysep 0, \mysep \ldots, \mysep 0 \big) $).
Namely, for any $ y \in \Delta_\ell^n $ we can write
\begin{equation}
 y = x + \sum_{i,j} \alpha_{i,j} f_{i,j}, 
\end{equation}
for some integers $ \alpha_{i,j} \geq 0 $. If $ d(x,y) = \delta $, then there
exists such a representation of $ y $ with $ \sum_{i,j} \alpha_{i,j} = \delta $.
We call two directions $ f_{i,j} $ and $ f_{k,l} $ orthogonal if
$ \{ i, j \} \cap \{ k, l \} = \emptyset $, i.e., if there is no coordinate
at which both of them are nonzero.

The following claim is a generalization of Lemma \ref{prop_orthogonal_2D}
to higher dimensions. Suppose we have two codewords ($ x, y $) and a point
$ w $ lying outside their decoding regions. The lemma asserts that if
$ w $ is bounded by $ { \mathcal B }(x, e) $and $ { \mathcal B }(y, e) $ in
some direction, say $ f_{1,2} $, then the codeword $ z $ covering $ w $ has
to lie in the subspace orthogonal to $ f_{1,2} $, i.e., it must be of the
form
\begin{equation}
 z = w + \big( 0, \mysep  0 , \mysep s_2, \mysep \ldots , \mysep s_n \big)
\end{equation}
where $ \sum_{i} s_i = 0 $ and  $ \sum_{i} |s_i| = 2e $. 
\begin{lemma}
\label{prop_orthogonal}
 Let $ x, y, w \in \Delta_\ell^n $ be such that $ d(x, w) = d(y, w) = e + 1 $,
$ d( x, w + f_{1,2} ) = e $, and $ d( y, w + f_{2,1} ) = e $. Then the point
$ z $ such that $ w \in { \mathcal B }(z,e) $,
$ { \mathcal B }(x,e) \cap { \mathcal B }(z,e) = \emptyset $ and
$ { \mathcal B }(y,e) \cap { \mathcal B }(z,e) = \emptyset $ must have a
representation of the form:
\begin{equation}
 z = w + \sum_{i,j \notin \{1, 2\}} \alpha_{i,j} f_{i,j}  ,
\end{equation}
with $ \alpha_{i,j} \geq 0 $, $ \sum_{i,j \notin \{1,2\}} \alpha_{i,j} = e $.
\end{lemma}
\begin{proof}
 The point $ z $ has to be at distance $ e $ from $ w $. (If the distance
were larger, the ball $ { \mathcal B }(z,e) $ would not contain $ w $, and
if it were smaller this ball would intersect $ { \mathcal B }(x,e) $ and
$ { \mathcal B }(y,e) $.)
We can therefore write
\begin{equation}
 z = w + \sum_{i,j} \alpha_{i,j} f_{i,j} 
\end{equation}
where $ \alpha_{i,j} \geq 0 $, $ \sum_{i,j} \alpha_{i,j} = e $. We need
to show that in such a representation we necessarily have $ \alpha_{i,j} = 0 $
whenever $ i \in \{1, 2\} $ or $ j \in \{1, 2\} $. Suppose that this is
not true, and that $ \alpha_{1,3} > 0 $ for example (the proof is similar
if any other $ \alpha_{i,j} $ with $ i \in \{1, 2\} $ or $ j \in \{1, 2\} $
is assumed positive). Since $ f_{1,3} = f_{1,2} + f_{2,3} $, we can write
\begin{equation}
 \begin{aligned}
  z &= w + f_{1,2} + f_{2,3} + (\alpha_{1,3} - 1) f_{1,3}
         + \sum_{(i,j)\neq(1,3)} \alpha_{i,j} f_{i,j}   \\
    &= w + f_{1,2} + \sum_{i,j} \beta_{i,j} f_{i,j} ,
 \end{aligned}
\end{equation}
where $ \beta_{i,j} \geq 0 $, $ \sum_{i,j} \beta_{i,j} = e $, which
implies that $ d(z, w + f_{1,2}) = e $. But we have assumed that also
$ d(x, w + f_{1,2}) = e $, which means that
$ { \mathcal B }(x,e) \cap { \mathcal B }(z,e) \neq \emptyset $,
a contradiction.
\end{proof}
\begin{remark}
\label{rem_orthogonal_2D}
 Since there are no orthogonal directions in the two-dimensional simplex
$ \Delta_\ell^2 $, the above lemma implies that if $ w $ is
``trapped'' between $ { \mathcal B }(x,e) $ and $ { \mathcal B }(y,e) $,
then there exists no $ z $ with $ w \in { \mathcal B }(z,e) $ and 
$ { \mathcal B }(z,e) \cap { \mathcal B }(x,e) = 
{ \mathcal B }(z,e) \cap { \mathcal B }(y,e) = \emptyset $.
This is precisely the statement of Lemma \ref{prop_orthogonal_2D}.
\end{remark}
\par Let us now continue with the proof of nonexistence of perfect codes.
As in the two-dimensional case, we start by observing the vertex
$ \big( \ell , \mysep  0 , \mysep  \ldots , \mysep  0 \big) $.
For this vertex to be covered there must exist a codeword of the form
\begin{equation}
\label{eq_xgeneral}
 x =  \big( \ell - t , \mysep  x_1 , \mysep  \ldots , \mysep  x_n \big) 
\end{equation}
with $ x_1 + \ldots + x_n = t \leq e $.
Without loss of generality, we assume that $ x_1 > 0 $ whenever $ t > 0 $.
Observe now the point
\begin{equation}
 v =  \big( \ell - x_1 - e - 1 , \mysep  x_1 + e + 1 ,
            \mysep  0 , \mysep  \ldots , \mysep  0
      \big)
\end{equation}
We have $ d(x,v) = e + 1 $ and so the point $ v $ is not covered by
$ { \mathcal B }(x,e) $. To cover it we need another codeword $ y $
with $ d(v,y) = e $ and $ d(x,y) = 2e + 1 $.
\begin{lemma}
\label{prop_ygeneral}
 The point $ y $ satisfying $ d(v,y) = e $, $ d(x,y) = 2e+1 $ is of the
form
\begin{equation}
\label{eq_ygeneral}
 y =  \big( \ell - x_1 - 2e - 1 , \mysep  x_1 + e + 1 + u ,
            \mysep  y_2 , \mysep  \ldots , \mysep  y_n
      \big)
\end{equation}
with $ 0 \leq u \leq e $, $ y_2 + \cdots + y_n = e - u $, and with the
property that
\begin{equation}
\label{eq_yxzero}
 x_i > 0  \Rightarrow  y_i = 0 \quad \text{for} \quad i = 2, \ldots, n . 
\end{equation}
\end{lemma}
\begin{proof}
 Let 
$ y = \big( \ell - x_1 - 2e - 1 + s , \mysep  y_1 , \mysep  \ldots , \mysep  y_n \big) $
for some $ s \in \mathbb{Z} $. If $ s < 0 $ we have
$ d(v,y) \geq v_0 - y_0 = e - s > e $ which contradicts one of the
assumptions of the lemma. Let us show that the case $ s > 0 $ is
also impossible. We can assume that $ x_0 > y_0 $; otherwise, the
vertex $ \big( \ell , \mysep 0 , \mysep  \ldots , \mysep 0 \big) $
would be covered by both $ x $ and $ y $. We can also assume that
$ s \leq x_1 $, for otherwise we would have $ x_0 - y_0 \leq 2e - t $,
and since the sum of the remaining $ x_i $'s is $ t $ it would follow
that
\begin{equation}
\begin{aligned}
 d(x,y) &=     \sum_{ x_i > y_i }  ( x_i - y_i )
         =     x_0 - y_0 + \sum_{ i > 0,\, x_i > y_i }  ( x_i - y_i )  \\
        &\leq  x_0 - y_0 + \sum_{ i > 0 }  x_i
         \leq   2e .
\end{aligned}
\end{equation}
Since $ v_0 - y_0 = e - s < e $ and $ y_i \geq v_i = 0 $ for $ i \geq 2 $,
we must have $ v_1 - y_1 = x_1 + e + 1 - y_1 = s $ in order to achieve
$ d(v,y) = e $, and hence
\begin{equation}
\label{eq_y1x1}
 y_1 = x_1 - s + e + 1 \geq e + 1 > x_1 , 
\end{equation}
where the first inequality follows from the above assumption that
$ s \leq x_1 $. Since $ y_0 < x_0 $ and $ y_1 - x_1 = e + 1 - s $,
in order to have $ d(x,y) = 2e + 1 $ some of the remaining $ y_i $'s,
$ i \geq 2 $, have to be greater than the corresponding $ x_i $'s
for exactly $ \sum_{ i \geq 2,\, y_i > x_i } ( y_i - x_i ) = e + s $.
But this is impossible because
\begin{equation}
 \sum_{ i \geq 2,\, y_i > x_i } ( y_i - x_i ) \leq 
 \sum_{ i \geq 2 } y_i  =  \ell - y_0 - y_1  =  e  
 <  e + s , 
\end{equation}
where we have used \eqref{eq_y1x1}. We thus conclude that $ s $
must be zero. In that case we have $ v_0 - y_0 = e $, and since
$ d(v,y) = e $, we must also have $ y_1 \geq v_1 = x_1 + e + 1 $.
This shows that $ y $ is necessarily of the form \eqref{eq_ygeneral}.
To prove the last part of the claim observe that $ y_0 < x_0 $, 
$ y_1 - x_1 = e + 1 + u $, and $ d(x,y) = 2e + 1 $ imply that
$ \sum_{ i \geq 2,\, y_i > x_i } ( y_i - x_i ) = e - u $.
But since $ \sum_{ i \geq 2 } y_i = e - u $, this can
only hold if $ x_i = 0 $ whenever $ y_i > 0 $, $ i \geq 2 $.
\end{proof}
\par Assume therefore that we have two codewords of the form
\eqref{eq_xgeneral} and \eqref{eq_ygeneral}, and observe the point
\begin{equation}
 \begin{aligned}
  w = \big( &\ell - t - e - 1 , \mysep  x_1 + u , \\
             & \max\{x_2, y_2\} + 1 , \mysep  \max\{x_3, y_3\}, 
               \mysep  \ldots, \mysep  \max\{x_n, y_n\}
      \big) . 
 \end{aligned}
\end{equation}
By using \eqref{eq_yxzero} it is not hard to conclude that
$ w \in \Delta_\ell^n $ and that $ d(x,w) = d(y,w) = e + 1 $, and hence
we need a third codeword $ z $ to cover $ w $. Such a codeword, however,
cannot exist, as shown below.

Assume first that $ u > 0 $.
Then we have that $ d(x, w + f_{1,2}) = e $ and $ d(y, w + f_{2,1})  = e $.
By using Lemma \ref{prop_orthogonal} we then conclude that the
codeword $ z $ which covers $ w $ must be of the form
$ z = w + \big( 0, \mysep 0, \mysep s_2, \mysep \ldots, \mysep s_n \big) $
with $ \sum_{i} s_i = 0 $ and $ \sum_{i} |s_i| = 2e $ (the second condition
is needed in order to have $ d(z, w) = e $). Therefore
\begin{equation}
\begin{aligned}
 z = \big(  &\ell - t - e - 1 , \mysep  x_1 + u ,  
            \mysep \max\{x_2, y_2\} + 1 + s_2 , \\ & \max\{x_3, y_3\} + s_3 , 
             \mysep   \ldots , \mysep   \max\{x_n, y_n\} + s_n
     \big)
\end{aligned}
\end{equation}
Now, since $ x_0 - z_0 = e + 1 $ and $ x_1 < z_1 $ we must have
\begin{equation}
\label{eq_zfromx}
 \sum_{ i \geq 2,\, x_i > z_i } ( x_i - z_i )  =  e 
\end{equation}
in order for $ d(x,z) = 2e + 1 $ to hold. Similarly, from $ z_0 > y_0 $
and $ y_1 - z_1 = e + 1 $ we conclude that
\begin{equation}
\label{eq_zfromy}
 \sum_{ i \geq 2,\, y_i > z_i } ( y_i - z_i )  =  e . 
\end{equation}
But it is not hard to conclude that we cannot simultaneously have
\eqref{eq_zfromx} and \eqref{eq_zfromy} because $ x_i $'s and $ y_i $'s,
$ i \geq 2 $, are never simultaneously positive \eqref{eq_yxzero}. Namely,
since $ \sum_{ s_i < 0 } |s_i| = e $, even if we achieve $ d(y,z) = 2e + 1 $
(by letting $ s_i $'s to be negative on the coordinates where $ y_i $'s
are positive), we would have $ d(x,z) = e + 1 $ because there are no more
negative $ s_i $'s to obtain \eqref{eq_zfromx}. We thus conclude that it
is not possible to find a codeword $ z $ which covers $ w $, and whose
decoding region is disjoint from those of the codewords $ x $ and $ y $.

It is left to consider the case when $ u = 0 $. In that case
\begin{equation}
 y = \big( \ell - x_1 - 2e - 1 , \mysep  x_1 + e + 1 , 
           \mysep  y_2 , \mysep  \ldots , \mysep  y_n \big) . 
\end{equation}
Note that now $ y_2 + \cdots + y_n = e $ and hence we can assume that
$ y_2 > 0 $. Observe the point
\begin{equation}
\begin{aligned}
 w' = \big( &\ell - t - e - 1 , \mysep  x_1 + 1 , \mysep  y_2 - 1 , \\
            &\max\{x_3, y_3\} + 1 , \mysep  \max\{x_4, y_4\} , 
              \mysep  \ldots , \mysep  \max\{x_n, y_n\} 
      \big) . 
\end{aligned}
\end{equation}
We again have $ d(x,w') = d(y,w') = e + 1 $, and $ d(x, w' + f_{1,2}) = 
d(y, w' + f_{2,1}) = e $. Therefore, the codeword $ z' $ covering $ w' $
has to be of the form
$ z' = w' + \big( 0, \mysep 0, \mysep r_2, \mysep \ldots, \mysep r_n \big) $
with $ \sum_{i} r_i = 0 $ and $ \sum_{i} |r_i| = 2e $. By the same reasoning
as above we conclude that we cannot simultaneously achieve that
$ d(x,z') = 2e + 1 $ and $ d(y,z') = 2e + 1 $, and hence the codeword $ z $
whose decoding region contains $ w' $ and is disjoint from the decoding
regions of $ x $ and $ y $ does not exist.

The proof of the claim is now complete -- nontrivial perfect codes in
$ \Delta_\ell^n $, $ n > 2 $, do not exist.
\section*{Acknowledgments}

The authors are very grateful to the reviewers for a detailed reading
and many useful comments on the original version of the manuscript.
This work was supported by the Ministry of Science and Technological
Development of the Republic of Serbia (grants TR32040 and III44003).
Part of the work was done while M. Kova\v cevi\'c was visiting Aalborg
University, Denmark, under the support of the COST action IC1104.
He is very grateful to the Department of Electronic Systems, and in
particular to \v Cedomir Stefanovi\'c and Petar Popovski, for their
hospitality.
\end{document}